\documentclass[a4paper,11pt]{amsart}%
\usepackage[foot]{amsaddr}

\usepackage{graphicx}        
\usepackage{multicol}        
\usepackage[justification=centering]{caption}
\usepackage{microtype}
\usepackage{algorithm}
\usepackage{siunitx}
\usepackage{thmtools} 
\usepackage[noend]{algorithmic}

\usepackage{amssymb}
\usepackage{amsmath}
\usepackage{mathtools}
\usepackage{amsfonts}
\usepackage{booktabs}
\usepackage{multirow}
\usepackage{exscale}
\usepackage{latexsym}
\usepackage{xspace}
\usepackage{enumerate}
\usepackage{todonotes}
\usepackage{hyperref}
\usepackage{cleveref}

\usepackage{color}

\usepackage[inline]{enumitem}

\newlist{cvdesc}{description}{1}
\setlist[cvdesc]{nosep,
labelindent=0pt,
labelwidth=2.8cm,
labelsep*=0.2cm,
leftmargin=cm,
font=\normalfont,
align=right}
\newlist{compactenum}{enumerate}{3}
\setlist[compactenum]{nosep,itemsep=2pt,topsep=2pt,labelindent=1em,leftmargin=1em}
\setlist[compactenum,1]{label=\textbullet}

\newtheorem{theorem}{Theorem}
\newtheorem{lemma}[theorem]{Lemma}
\newtheorem{proposition}[theorem]{Proposition}
\newtheorem{corollary}[theorem]{Corollary}
\newtheorem{remark}[theorem]{Remark}
\newtheorem{definition}[theorem]{Definition}
\newtheorem{assumption}[theorem]{Assumption}

\newtheorem{example}[theorem]{Example}

\newcommand{\fieldx}{\ensuremath{K}}

\newcommand{\N}{\mathbb{N}}
\newcommand{\Q}{\mathbb{Q}}
\newcommand{\Z}{\mathbb{Z}}

\DeclareMathOperator{\lm}{lm}
\DeclareMathOperator{\LM}{lm}
\DeclareMathOperator{\LD}{LD}

\def\<#1>{\langle#1\rangle}

\newcommand{\amba}{\mathfrak{a}}

\DeclareMathOperator{\supp}{supp}

\newcommand{\ideal}{\unlhd}

\def\Fr{\Sigma}
\def\MFr{M(\Fr)}
\DeclareMathOperator{\sig}{sig}
\DeclareMathOperator{\siga}{SIG}
\DeclareMathOperator{\lma}{LM}
\DeclareMathOperator{\Syz}{Syz}
\def\Sig{^{[\Sigma]}}

\begin{document}

\title{Modular Algorithms For Computing Gr\"obner Bases in Free Algebras}
\author[C.~Hofstadler]{Clemens Hofstadler\textsuperscript{1}}
\renewcommand\thefootnote{$\ast$}
\address{\textsuperscript{1}Institute for Symbolic Artificial Intelligence,
        Johannes Kepler University, Linz, Austria}
\email{clemens.hofstadler@jku.at}

\author[V. Levandovskyy]{Viktor Levandovskyy\textsuperscript{2}}
\address{\textsuperscript{2}American University Kyiv, Ukraine}
\email{viktor.levandovskyy@auk.edu.ua}

%
\begin{abstract}
In this work, we extend modular techniques for computing Gröbner bases involving rational coefficients to (two-sided) ideals in free algebras.
We show that the infinite nature of Gröbner bases in this setting renders the classical approach infeasible.
Therefore, we propose a new method that relies on signature-based algorithms.
Using the data of signatures, we can overcome the limitations of the classical approach and obtain a practical modular algorithm.
Moreover, the final verification test in this setting is both more general and more efficient than the classical one.
We provide a first implementation of our modular algorithm in \textsc{SageMath}.
Initial experiments show that the new algorithm can yield significant speedups over the non-modular approach.
\end{abstract}

\maketitle

%

\section{Introduction}

Gröbner bases are a fundamental tool in computational algebra and algebraic geometry, providing a
systematic approach to solving a wide range of problems related to  polynomial ideals and algebraic varieties.
Initially introduced for commutative polynomials~\cite{Buc65}, the theory of Gröbner bases has since been extended to various noncommutative settings, including $G$-algebras~\cite{KW} (see also~\cite{vikThesis} and references therein) and free (associative) algebras~\cite{Ber78,Mor85}. 
In the latter settings, Gröbner bases have found important applications in studying finitely presented algebras or in automated theorem proving~\cite{HW94,CHRR20,SL20,RRHP21,Hof23}.

When computing Gröbner bases over the rationals, it is a well-known problem that intermediate coefficient
swell can drastically hamper computations. 
To address this issue, modular methods have been developed, which first compute a Gröbner basis of a given
ideal modulo different primes, thereby bounding the size of the coefficients, and then reconstruct
a rational Gröbner basis from these modular images.
Originally proposed for commutative polynomials~\cite{elbert83, A, pfister2011}, these techniques have since been extended to $G$-algebras~\cite{DEVT20}.

In this work, we develop a modular Gröbner basis algorithm for (two-sided) ideals in free algebras.
In contrast to all previously studied settings, the free algebra is not Noetherian, which implies that Gröbner bases can be infinite.
As a result, algorithms for computing them may not terminate and serve only as enumeration procedures.
The traditional approach to modular Gröbner basis computations relies implicitly on the assumption that all Gröbner bases are finite.
Since this is no longer guaranteed in the free algebra, the classical approach becomes infeasible.

In particular, a key aspect of modular algorithms is their dependence on so-called \emph{lucky} primes. 
Only from such primes a correct result can be reconstructed. 
In the classical settings of modular Gröbner basis algorithms, it can be proven that for any given ideal, there are only finitely many unlucky primes~\cite{A,DEVT20}.
In contrast, we show in Example~\ref{ex:infinitely-many-unlucky-primes} that there exist ideals in the free algebra for which every prime is unlucky.
Moreover, in our setting, measuring progress in the computation of an infinite Gröbner basis poses a significant challenge.
However, addressing this issue is critical for the development of a modular algorithm, as it directly affects the reconstruction of the rational basis and its subsequent verification.
We discuss these problems in more detail in Section~\ref{sec:problems-with-classical-approach}.

To remedy the issues mentioned above, we propose a new approach to modular Gröbner basis computations in which we replace traditional Gröbner basis algorithms by signature-based methods.
Since their introduction with the goal of improving Gröbner basis computations for commutative polynomials~\cite{Fau02}, 
signature-based algorithms have been studied extensively, see~\cite{EF17,Lai24} and references therein.
Recently, one of the authors has generalized these algorithms to the free algebra in joint work~\cite{HV22,HV23}.

Signature-based algorithms compute, in addition to a Gröbner basis, some information on how the polynomials in that basis were derived from the input.
This additional information, called \emph{signatures}, not only allows the algorithms to predict and avoid redundant computations, 
but it can also be exploited in other applications, for instance, to speed up certain ideal\nobreakdash-theoretic operations~\cite{ELMS23},
to compute information related to minimal projective resolutions~\cite{King14},
or to compute short certificates for ideal membership~\cite{HV24}.
Moreover, signature-based algorithms can enumerate a (potentially infinite) Gröbner basis in an incremental fashion, 
analogous to a degree-by-degree computation of a homogeneous Gröbner basis. 
Precisely this structure is what we exploit in our modular algorithm.

Another advantage of our signature-based approach over the classical method is a more general and efficient verification criterion.
In a modular algorithm, it is crucial to verify whether the reconstructed result is indeed correct.
In the case of Gröbner basis computations, this involves two steps: 
\begin{enumerate*}
\item verifying that the reconstructed basis generates the same ideal as the original input, and\label{item:intro-1}
\item verifying that the reconstructed basis is a Gröbner basis.\label{item:intro-2}
\end{enumerate*}
Using the classical approach,~\eqref{item:intro-1} is in other noncommutative settings only possible for homogeneous ideals, cf.~\cite[Thm.~22]{DEVT20}, and~\eqref{item:intro-2} relies on Buchberger's S-polynomial criterion, which requires a substantial amount of polynomial arithmetic, in general, quadratic in the size of the candidate basis.
In contrast, our signature-based approach utilizes the \emph{cover criterion}~\cite{GVW16,HV23}, which yields a verification procedure that is applicable to all ideals, including inhomogeneous ones, and that reduces the cost of expensive polynomial arithmetic to a 
linear amount in the size of the candidate basis.


We have implemented our modular algorithm (Algorithm~\ref{modGalgo}) as part of the \textsc{SageMath} package \texttt{signature\_gb}
\footnote{available at \url{https://github.com/ClemensHofstadler/signature_gb}}~\cite{Hof23}, 
which offers functionality for signature-based Gröbner basis computations in free algebras.
In Section~\ref{sec:experiments}, we report on the performance of our new algorithm and show that it can drastically speed up the computation of noncommutative Gröbner bases.

\section{Preliminaries}\label{prelim}

In this section, we recall basic notions and constructions relevant for this work.
In particular, we recall the most important aspects of the Gröbner basis theory in free algebras. 
For proofs and further details, we refer readers to, e.g.,~\cite{bk07,Mor16,Hof23}.

\subsection{Free Algebra}

Given a finite set of indeterminates $X = \{x_1,\ldots,x_n \}$, 
we write ${\left\langle X\right\rangle} = \langle x_{1},\ldots, x_{n} \rangle$ for the \emph{free monoid} on
$X$, containing all finite words over the alphabet $X$ including the empty word, which we identify with $1$.
We refer to elements in $\langle X \rangle$ as \emph{monomials}.

For a field $K$, we consider the corresponding monoid algebra $K\langle X\rangle = \fieldx\langle
x_{1},\ldots, x_{n} \rangle$, that is, the \emph{free (associative) $K$-algebra} generated by $\langle X \rangle$.
Elements in $K\<X>$ are finite sums of the form $c_1 w_1 + \cdots + c_d w_d$ with nonzero $c_1,\dots,c_d \in K$ and pairwise distinct $w_1,\dots,w_d \in \<X>$.
We consider these elements as \emph{(noncommutative) polynomials} with coefficients in $K$ and monomials in $\<X>$.
The \emph{support} of $f = c_1 w_1 + \cdots + c_d w_d \in K\<X>$ is
$$
    \supp(f)=\{ w_1,\dots,w_d\}.
$$
By convention, for the zero polynomial, we have $\supp(0)=\emptyset$.

In this work, when speaking about an ideal $I \subseteq K\<X>$, we always mean a two-sided ideal, also denoted by $I \ideal K\<X>$.
Given a set of polynomials $F \subseteq K\<X>$, we let $\<F>$ be the ideal generated by $F$, that is,
$$
	\<F> = \left\{\sum_{i=1}^k p_i f_i q_i \mid f_i\in F,\; p_i, q_i\in K\<X>,\; k\in\N\right\}.
$$
The set $F$ is called a \emph{generating set} of $\<F>$.
An ideal $I \ideal K\<X>$ is \emph{finitely generated} if it has a finite generating set.

\begin{remark}
If $|X| > 1$, then there exist ideals in $K\<X>$ which are not finitely generated.
A classical example is given by $\< \{ xy^n x \mid n \in \N \}> \ideal K\<X>$.
\end{remark}

A total ordering $\prec$ on $\left\langle X\right\rangle$ is called a 
\emph{monomial (well-)ordering} if the following conditions hold:
\begin{enumerate}
        \item $\forall a,b,v,w \in \left\langle X\right\rangle$: if $v\prec w$, then $a \cdot v \cdot b \prec a\cdot w \cdot b$;
        \item every nonempty subset of $\<X>$ has a least element.\label{cond:monomial-order-2}
\end{enumerate}

\begin{remark}
By Higman's lemma~\cite{Hig52}, condition~\eqref{cond:monomial-order-2} is equivalent to
\begin{enumerate}
    \item[\textnormal{(2')}] $\forall w \in \left\langle X\right\rangle$: if $w \neq 1$, then $1 \prec w$.
\end{enumerate}
However, note that this equivalence only holds for finite $X$.
\end{remark}

\begin{remark}
Recall, that the left (and, analogously, right) lexicographical orderings with, say $x \prec y$, are not well-orderings: $1 \prec x$, multiplied by $y$ from the right implies $y \prec xy$. 
However, since $y \succ x$ holds, comparison from the left delivers $y \succ xy$, 
which is a contradiction.
Nonetheless, using a lexicographic ordering as a tie-breaker after first comparing the (weighted) degrees of different words yields classical examples of monomial orderings.
\end{remark}

For a fixed monomial ordering $\prec$, we define the \emph{leading monomial} of a nonzero polynomial $f\in K\<X>$ as
$$
    \LM(f) = \max_\prec \supp(f) \in \<X>.
$$
If the coefficient of $\lm(f)$ in $f$ is $1$, then $f$ is called \emph{monic}.

So far, the leading monomial of the zero polynomial is undefined.
To avoid exceptions for the zero polynomial,
it is convenient to extend a monomial ordering $\prec$ to $\<X> \cup \{0\}$ and set $0 \prec w$ for all $w \in \<X>$.
In the following, we consider such an extension, and with this in mind, we define $\LM(0) = 0$.

\subsection{Gr\"obner Bases in Free Algebras}  
\label{sec:GB}

Given a monomial ordering $\prec$ on~$\<X>$ and a subset $I\subseteq K\<X>$, we set
$$
    \LM(I)= \left\{\LM(f)\mid f\in I \setminus \{0\}\right\} \subseteq \langle X \rangle.
$$
Note that if $I$ is a nonzero ideal of $K\<X>$, then $\LM(I)$ is a (two-sided) monoid ideal of the free monoid $\langle X \rangle$.  


\begin{definition} Let $I \ideal K\<X>$ be a nonzero ideal. 
A \emph{(two-sided) Gr\"obner basis} of $I$ (with respect to $\prec$) is a subset  
$\mathcal G\subseteq I$ such that 
$$
    \<\LM(I)> = \<\LM(\mathcal G)>.
$$
A subset  $\mathcal G\subseteq K\<X>$ is a \emph{Gr\"obner basis} 
if it is a Gr\"obner basis of the ideal $\langle \mathcal G\rangle\ideal K\<X>$ it generates.
\end{definition} 





Contrary to the case of Noetherian algebras, not all ideals in $K\<X>$ possess a finite Gr\"obner basis. 
Nevertheless, there exists an adaptation of Buchberger's S-polynomial criterion to free algebras, known as 
\emph{Bergman's diamond lemma}~\cite[Thm.~1.2]{Ber78}, and based on that, a noncommutative analogue of Buchberger's algorithm to enumerate
a (possibly infinite) Gröbner basis~\cite{Mor85}.
This procedure terminates if and only if the ideal admits a finite Gröbner basis (w.r.t.~the chosen monomial ordering).

The notion of a reduced Gr\"obner basis is analogous to that in the commutative case and, 
for a fixed monomial ordering, every ideal of $K\<X>$ has a uniquely determined such reduced Gr\"obner basis~\cite[Prop.~2.4.41]{Hof23}.

\section{Signature-based Gröbner Bases in Free Algebras}
\label{sec:signature-gb}

Our modular algorithm will rely crucially on properties of noncommutative signature-based Gröbner bases and the corresponding algorithms to compute them.
We therefore recall their theory and the relevant results in the following.
For a more in-depth discussion, see~\cite[Ch.~3]{Hof23}.

\subsection{Free Bimodule}

Signature-based algorithms rely on adding a module perspective to the polynomial computations.
Therefore, for fixed $r \in \N$, we let $\Sigma = (K\<X> \otimes K\<X>)^r$ be the \emph{free $K\<X>$-bimodule} of rank $r$.
If we denote $\varepsilon_i = (0,\dots,0, 1\otimes 1, 0, \dots,0)$ with $1 \otimes 1$ appearing in the $i$th position ($i = 1,\dots,r$), then
the canonical basis of $\Sigma$ is given by $\varepsilon_1,\dots,\varepsilon_r$.

A \emph{bimodule monomial} is an expression of the form $a \varepsilon_i b$ with $a,b \in \<X>$ and $1 \leq i\leq r$.
The set of all bimodule monomials is denoted by $\MFr$.
Note that $\Sigma$ forms a $K$-vector space with basis $\MFr$,
that is, every $\alpha \in \Sigma$ can be written uniquely as $\alpha = c_1 \mu_1 + \dots + c_d \mu_d$
with nonzero $c_1,\dots,c_d \in K$ and pairwise distinct $\mu_1,\dots,\mu_d \in \MFr$.
The support of such an element is 
$$
    \supp(\alpha) = \{\mu_1,\dots,\mu_d\}.
$$

A total ordering $\prec_\Sigma$ on $\MFr$ is called a 
\emph{bimodule (well-)ordering} if the following conditions hold:
\begin{enumerate}        
\item $\forall \mu,\sigma \in\MFr, 
\forall v,w \in \left\langle X\right\rangle$: 
if $\mu\prec_\Sigma \sigma$, then $v\cdot \mu \cdot w \prec_\Sigma v\cdot \sigma \cdot w$;
\item every nonempty subset of $\MFr$ has a least element.\label{item:bimodule-ordering}
\end{enumerate}

Many classical bimodule orderings arise as extensions of monomial orderings, see \cite{bk07, Mor16}. 
A property of bimodule orderings that will be crucial for our modular Gr\"obner basis algorithm is that of fairness.

\begin{definition}
A bimodule ordering $\prec_\Sigma$ is \emph{fair} if the set 
$\{ \mu \in \MFr \mid \mu \prec_\Sigma \sigma \}$
is finite for all $\sigma \in \MFr$.
\end{definition}

For $K\<X>$ itself (i.e., when $r=1$) the name \emph{well-founded} has also been used in the literature.

\begin{example}
Let us quickly review a recipe for establishing a fair monomial ordering on $K\<X>$. 
Suppose that every variable $x_i \in \<X>$ is assigned a positive \emph{weight} $\omega_i\in\mathbb{R}_{> 0}$. 
We define the \emph{weighted degree} function $\deg_\omega \colon \<X> \to \mathbb{R}$ by setting $\deg_\omega(x_{i_1}\dots x_{i_d}) = \sum_{j=1}^d \omega_{i_j}$, with the special case $\deg_\omega(1)=0$.
Classical \emph{degree} orderings, such as degree lexicographic, assign all weights to~$1$.

Such a weighted degree function $\deg_\omega$ can extend any tie-breaking relation $\prec$ (like left or right lexicographic comparisons, which are not monomial orderings themselves) on monomials of the same weighted degree to a monomial ordering $\prec_\omega$. 
Such an ordering is always fair 
and can be used to construct a family of fair bimodule orderings as explained below. 
\end{example}

\begin{proposition}
\label{prop:fair-ordering}
Let $\prec$ be a fair monomial ordering on $K\<X>$ and $<$ be a strict total ordering on the generators $\varepsilon_1,\ldots,\varepsilon_r$ of the free bimodule $\Fr$. 
Define a new relation $\prec_\Fr$ on $\MFr$ as follows: 
\renewcommand{\thefootnote}{\fnsymbol{footnote}}
\begin{align*}
a\varepsilon_i b \prec_\Fr
c \varepsilon_j d \iff &\left( ab \prec cd \right) \; \vee \\
 &\left(ab = cd \wedge \varepsilon_i < \varepsilon_j \right)\; \vee \\
 &\left(ab = cd \wedge \varepsilon_i=\varepsilon_j \wedge a \prec c\right)\footnotemark.
\end{align*}
\footnotetext{replacing $a \prec c$ by $b \prec d$ yields a variant where we compare from the right}
Then $\prec_\Fr$ is a fair bimodule ordering.
\end{proposition}

\begin{proof}
Let $v,w \in \<X>$ and suppose that $a\varepsilon_i b \prec_\Fr c \varepsilon_j d$ holds. We first have to show that $v \cdot a\varepsilon_i b \cdot w \prec_\Fr
v\cdot c \varepsilon_j d \cdot w$ holds. 
Indeed, by the definition, we have to prove that the following is true:
\begin{gather*}
\left( vabw \prec vcdw \right) \; \vee \;
 \left(vabw = vcdw \wedge \varepsilon_i < \varepsilon_j \right) \; \vee \\
 \left(vabw = vcdw \wedge  \varepsilon_i = \varepsilon_j \wedge va \prec vc \right),
\end{gather*}
but this follows from the fact that $\prec$ is a monomial ordering.
Clearly, it also holds for the variant where we replace $a \prec c$ by $ b \prec d$, 
since $b \prec d$ implies $bw \prec dw$.

Note that fairness of a total ordering implies Condition~\eqref{item:bimodule-ordering} of a bimodule ordering.
Thus, it remains to prove that $\prec_\Fr$ is fair.
To this end, let $c \varepsilon_j d \in \MFr$.
We analyze whether the set 
\begin{align*}
S = \left\{a\varepsilon_i b \in \MFr \mid a\varepsilon_i b \prec_\Fr c \varepsilon_j d \right\}
\end{align*}
is finite. 
First, note that there are only finitely many bimodule monomials $a \varepsilon_i b$
that satisfy the condition $ab \prec cd$, because $\prec$ is assumed to be fair.
Further, there are also only finitely many $a \varepsilon_i b$ that satisfy the condition $ab = cd \wedge \varepsilon_i < \varepsilon_j$, since the monomial $cd$ only has finitely many factorizations and there are only finitely many generators $\varepsilon_i$, which are smaller than $\varepsilon_j$.
Lastly, there are also only finitely many $a \varepsilon_i b$ that satisfy the last condition $ab = cd \wedge \varepsilon_i = \varepsilon_j \wedge a \prec c$.
All in all, there are only finitely many bimodule monomials that satisfy any of the conditions in the proposition, and therefore, the set $S$ is finite.
\end{proof}

The bimodule orderings constructed in Proposition~\ref{prop:fair-ordering} can be seen as generalizations of the concept of \emph{term-over-position ordering} $\prec_\textup{ToP}$, since the comparison of bimodule components occurs after the comparison of monomials. 

\begin{example}
The classical counterpart of a $\prec_\textup{ToP}$ ordering is the \emph{position-over-term ordering} $\prec_\textup{PoT}$, which extends a monomial ordering $\prec$ on $K\<X>$ as
\[
a\varepsilon_i b \prec_\textup{PoT} 
c \varepsilon_j d \iff
\left(\varepsilon_i < \varepsilon_j \right) \;\vee\; \left( \varepsilon_i=\varepsilon_j \,\wedge\, a\varepsilon_i b \prec 
c \varepsilon_i d\right),
\]
where the latter comparison  happens in the same component $\varepsilon_i$, and can be defined variously, for instance analogously to the previous Proposition~\ref{prop:fair-ordering}. 

An important observation is immediate: $\prec_\textup{PoT}$ is not fair if the rank of the bimodule $\Sigma$ is at least 2.
Indeed, the set $\{ \mu \in \MFr \mid \mu \prec_\textup{PoT} \varepsilon_2\}$ is infinite, 
containing all elements of the form $a \varepsilon_1 b$ with $a,b \in \<X>$.
\end{example}

\begin{example}



Also the bimodule components can be graded. 
Namely, with fixed nonzero $f_1,\dots,f_r \in K\<X>$, both bimodule orderings mentioned above can be extended to degree-compatible orderings by first comparing the weighted degrees of bimodule monomials, defined as $\deg_\omega(a \varepsilon_i b) \coloneqq \deg_\omega (a f_i b)$ for all $a \varepsilon_i b \in \MFr$, before using $\prec_\textup{PoT}$ or $\prec_\textup{ToP}$ as a tie-breaker.
This yields the degree-over-position-over-term ordering $\prec_\textup{DoPoT}$ and the 
degree-over-term-over-position ordering $\prec_\textup{DoToP}$, respectively. 
Note that, with this construction, the bimodule orderings $\prec_\textup{DoPoT}$ and $\prec_\textup{DoToP}$ are fair.
\end{example}

\emph{Signatures} are the module-equivalent to leading monomials of polynomials.
For the following, we fix a bimodule ordering $\prec_\Sigma$.

\begin{definition}
The \emph{signature} of a nonzero element $\alpha \in \Sigma$ is 
$$
    \sig(\alpha) = \max_{\prec_{\Sigma}} \supp(\alpha) \in \MFr.
$$
\end{definition}

Analogous to the case of polynomials, a nonzero bimodule element $\alpha$ is \emph{monic}
if the coefficient of $\sig(\alpha)$ in $\alpha$ is $1$.

\subsection{Signature-based Gröbner Bases in Free Algebras}

We fix a family of polynomials $(f_1,\dots,f_r) \in K\<X>^r$ generating an ideal $I = \<f_1,\dots,f_r>$.
To encode relations between the $f_i$, we consider the surjective $K\<X>$-bimodule homomorphism
$$
    	\overline{\cdot} \colon \Fr \to I, \quad \alpha = \sum_{i = 1}^d c_i a_i \varepsilon_{j_i} b_i \quad\mapsto\quad \overline  \alpha \coloneqq \sum_{i = 1}^d c_i a_i f_{j_i} b_i.
$$

A \emph{labeled polynomial} $f^{[\alpha]}$ is a pair $(f,\alpha) \in I \times \Sigma$ with $\overline{\alpha} = f$.
The set of all labeled polynomials is denoted by $I\Sig$, that is,
$$
    I\Sig = \{ f^{[\alpha]} \mid f \in I, \overline{\alpha} = f\}.
$$
Note that $I\Sig$ forms a $K\<X>$-subbimodule of $I \times \Sigma$ with component-wise addition and scalar multiplication.
We call it the \emph{labeled module} generated by $f_1,\dots,f_r$.
By a slight abuse of notation, we also write $I\Sig = \< f_1^{[\varepsilon_1]},\dots,  f_r^{[\varepsilon_r]}>$.

A \emph{syzygy} of $I\Sig$ is an element $\gamma \in \Sigma$ such that $\overline{\gamma} = 0$.
The set of all syzygies of $I\Sig$ is denoted by $\Syz(I\Sig)$. 
It forms a $K\<X>$-subbimodule of $\Sigma$.

A central concept for signature-based algorithms is that of \emph{sig-reducibility}.
We recall the definition in the following.

\begin{definition}\label{def:sig-reduction}
Let $f^{[\alpha]}, g^{[\beta]} \in I\Sig$ with $g \neq 0$.
We say that $f^{[\alpha]}$ is \emph{$\sig$\nobreakdash-reducible} by $g^{[\beta]}$ if there exist $a,b\in\<X>$ such that the following conditions hold:
\begin{enumerate}
	\item $\LM(agb) \in \supp(f)$;
	\item $\sig(a \beta b) \preceq \sig(\alpha)$;
\end{enumerate}

If $\lm(agb) = \lm(f)$, we say that $f^{[\alpha]}$ is \emph{top} sig-reducible.
Furthermore, we say that $f^{[\alpha]}$ is \emph{regular} sig-reducible if $\sig(a \beta b) \prec \sig(\alpha)$.

If $f^{[\alpha]}$ is (regular/top) sig-reducible by $g^{[\beta]}$, then the corresponding (regular/top) sig-reduction is
$$
    f^{[\alpha]} - (agb)^{[a\beta b]}.
$$
\end{definition}

The notions from Definition~\ref{def:sig-reduction} generalize to sets of labeled polynomials \mbox{$\mathcal G \subseteq I\Sig$} in a straightforward way.
In particular, $f^{[\alpha]}$ is (regular/top) $\sig$\nobreakdash-reducible by $\mathcal G$ if it is (regular/top) sig-reducible by an element in $\mathcal G$.
Otherwise, $f^{[\alpha]}$ is (regular/top) sig-reduced by $\mathcal G$.

Definition~\ref{def:sig-reduction} requires the reducer $g^{[\beta]}$ to have nonzero polynomial part.
In the following, however, we also need reductions by syzygies.
This leads to the following definition.

\begin{definition}
Let $f^{[\alpha]}, 0^{[\gamma]} \in I\Sig$ with $\gamma \neq 0$.
We say that $f^{[\alpha]}$ is \emph{syz\nobreakdash-reducible} by $0^{[\gamma]}$ if there exist $a,b\in\<X>$ such that $\sig(a \gamma b) \in \supp(\alpha)$.
If actually $\sig(a \gamma b) = \sig(\alpha)$, then $f^{[\alpha]}$ is \emph{top} syz-reducible by $0^{[\gamma]}$.
\end{definition}

The definition above extends to sets $\mathcal G \subseteq I\Sig$ as before:
$f^{[\alpha]}$ is (top) syz\nobreakdash-reducible by $\mathcal G$ if it is (top) syz-reducible by an element in $\mathcal G$.
Otherwise, it is (top) syz-reduced by $\mathcal G$.

Finally, we can recall the notion of \emph{strong Gröbner bases}\footnote{Note that the term \emph{strong Gröbner basis} is overloaded and also has a second meaning.
In other contexts, this term is used to describe special Gröbner bases in polynomial rings over Euclidean coefficient rings like $\Z$, cf.~\cite[Ch.~4]{AL94}.
}, which was first introduced in~\cite{GVW16} for commutative polynomials
and combines a Gröbner basis of $I$ and a Gröbner basis of $\Syz(I\Sig)$ in one set.

A key property of strong Gröbner bases, which we exploit in our modular algorithm is as follows. 
Strong Gröbner bases
can be characterized \emph{up to a signature} $\sigma \in \MFr$, 
meaning that the defining properties are satisfied  only by the elements with signature less than $\sigma$.
This is analogous to how a homogeneous Gröbner basis can be characterized \emph{up to a certain degree}.

\begin{definition}\label{def:strong-gb}
A subset $\mathcal G \subseteq I\Sig$ is a \emph{strong Gröbner basis} of $I\Sig$ \emph{(up to signature $\sigma \in \MFr$)} 
if the following conditions hold for all $f^{[\alpha]} \in I\Sig$ (with $\sig(\alpha) \prec \sigma$):
\begin{enumerate}
    \item if $f \neq 0$, then $f^{[\alpha]}$ is sig-reducible by $\mathcal G$;\label{cond:strong-gb-1}
    \item if $f = 0$, then $f^{[\alpha]}$ is syz-reducible by $\mathcal G$.\label{cond:strong-gb-2}
\end{enumerate}
\end{definition}

The first condition in Definition~\ref{def:strong-gb} implies that $\mathcal G$ contains a labeled Gröbner basis as introduced in~\cite[Def.~20]{HV22}
and, in particular, that the nonzero polynomial parts of $\mathcal G$ form a Gröbner basis of $I$.
The second condition implies that $\mathcal G$ contains also a Gröbner basis of the syzygy module $\Syz(I\Sig)$.
The following result, which is classical in the commutative setting~\cite[Prop.~2.2]{GVW16}, formalizes this observation.
We omit its proof here because the commutative proof carries over in a straightforward manner.

\begin{proposition}
Let $\mathcal G \subseteq I\Sig$ be a strong Gröbner basis of $I\Sig$.
Then the following hold:
\begin{enumerate}
    \item $\{ g \mid g^{[\beta]} \in \mathcal G, g \neq 0\}$ is a Gröbner basis of $I$;
    \item $\{ \beta \mid g^{[\beta]} \in \mathcal G, g = 0\}$ is a Gröbner basis of $\Syz(I\Sig)$.
\end{enumerate}
\end{proposition}

To end this section, we recall an important property of strong Gröbner bases (up to signature $\sigma$), which will be useful later.

\begin{lemma}\label{lemma:s-reductions}
Let $g^{[\beta]}, h^{[\gamma]}\in I\Sig$ be such that $\sig(\beta) = \sig(\gamma)$.
Furthermore, let $\mathcal G \subseteq I\Sig$ be a strong Gröbner basis up to signature $\sig(\beta)$.
If $g^{[\beta]}, h^{[\gamma]}$ are regular sig-reduced by $\mathcal G$, then $g = c \cdot h$ for some nonzero $c \in K$.
\end{lemma}

\begin{proof}
Consider $f^{[\alpha]} = g^{[\beta]} - c\cdot h^{[\gamma]}$,
with nonzero $c \in K$ such that the signatures cancel and we have $\sig(\alpha) \prec \sig(\beta)$.
Now assume, for contradiction, that $f \neq 0$.
By assumption $\sig(\alpha) \prec \sig(\beta)$.
Therefore, $f^{[\alpha]}$ is sig\nobreakdash-reducible by $\mathcal G$.
But this implies that $g^{[\beta]}$ or $h^{[\gamma]}$ is regular sig-reducible 
by $\mathcal G$ -- a contradiction to the assumption that both $g^{[\beta]}$ and $h^{[\gamma]}$ are regular sig\nobreakdash-reduced.
\end{proof}

\subsection{Reduced Strong Gröbner Bases}

In some applications it is beneficial to consider the reduced Gröbner basis of an ideal as it is unique.
Also our modular algorithm will rely crucially on the fact that we can characterize and compute a unique strong Gröbner basis.
Thus, in the following, we adapt the concept of reduced bases to our signature-based setting.

\begin{definition}\label{def:reduced}
A subset $\mathcal G \subseteq I\Sig$ is \emph{reduced}
if the following conditions hold for all $g^{[\beta]} \in \mathcal G$:
\begin{enumerate}
    \item if $g \neq 0$, then $g$ is monic; otherwise $\beta$ is monic;
    \item $g^{[\beta]}$ is regular sig-reduced by $\mathcal G$; 
    \item $g^{[\beta]}$ is top sig-reduced and syz-reduced by $\mathcal G \setminus \{g^{[\beta]}\}$.
\end{enumerate}
\end{definition}

Just as in the classical case, a reduced strong Gröbner basis is unique.
To show this, we introduce the following concept of \emph{leading data of a labeled polynomial}.
For $f^{[\alpha]} \in I\Sig$ with $\alpha \neq 0$, define its \emph{leading data} to be
$$
    \LD(f^{[\alpha]}) = \left(\LM(f), \sig(\alpha)\right) \in \left(\<X> \cup \{0\}\right) \times \MFr.
$$
Extend this to sets $\mathcal G \subseteq I\Sig$ by
$$
    \LD(\mathcal G) = \left\{ \LD(f^{[\alpha]}) \mid f^{[\alpha]} \in \mathcal G, \alpha \neq 0 \right\}.
$$
First, we show that the leading data of a reduced strong Gröbner basis is unique.

\begin{lemma}\label{lem:unique-leading-data}
If $\mathcal G$, $\mathcal H$ are reduced strong Gröbner bases of $I\Sig$, then $\LD(\mathcal G)= \LD(\mathcal H)$.
\end{lemma}

\begin{proof}
Let $g^{[\beta]} \in \mathcal G$.
We will show that there exists $h^{[\gamma]} \in \mathcal H$ such that $\LD(g^{[\beta]}) = \LD(h^{[\gamma]})$, implying that $\LD(\mathcal G) \subseteq \LD(\mathcal H)$.
The assertion of the lemma then follows by the symmetry of our argument.

We distinguish between two cases: $g = 0$ and $g \neq 0$.
If $g \neq 0$, then it is (top) sig-reducible by $\mathcal H$ because $\mathcal H$ is a strong Gröbner basis.
This means that there exist $h^{[\gamma]} \in \mathcal H$ and $a,b \in \<X>$ such that 
$$
    \lm(g) = \lm(ahb) \quad \text{ and } \quad \sig(\beta) \succeq \sig(a \gamma b).
$$
Since $\mathcal G$ is also a strong Gröbner basis, there exist $g'^{[\beta']} \in \mathcal G$ and $a',b' \in \<X>$ such that
$$
	 \lm(h) = \lm(a'g'b') \quad \text{ and }\quad  \sig(\gamma) \succeq \sig(a' \beta' b').
$$
Combining these two statements yields
\[
	\lm(g) = \lm(ahb) = \lm(aa' g' b' b) \quad \text{ and }\quad \sig(\beta) \succeq \sig(a \gamma b) \succeq \sig(a a' \beta' b' b).
\]
Now, if $g^{[\beta]} \neq g'^{[\beta']}$, then $g'^{[\beta']}$ could be used to top sig-reduce $g^{[\beta]}$ -- a contradiction to the fact that $\mathcal G$ is reduced.
So, $g^{[\beta]} = g'^{[\beta']}$, which implies that $a = a' = b = b' = 1$, and therefore, $\LD(g^{[\beta]}) = \LD(h^{[\gamma]})$.
The case $g = 0$ follows along the same lines, replacing top sig-reducibility by top syz\nobreakdash-reducibility. 
\end{proof}

Using Lemma~\ref{lem:unique-leading-data}, we can now prove the uniqueness of reduced strong Gröbner bases.

\begin{proposition}
If $\mathcal G$, $\mathcal H$ are reduced strong Gröbner bases of $I\Sig$, then $\mathcal G = \mathcal H$.
\end{proposition}

\begin{proof}
Lemma~\ref{lem:unique-leading-data} implies that $\LD(\mathcal G) = \LD(\mathcal H)$.
Then, let $g^{[\beta]} \in \mathcal G$ and $h^{[\gamma]} \in \mathcal H$ be such that $\lm(g) = \lm(h)$ and $\sig(\beta) = \sig(\gamma)$.

First, we show that $g = h$.
Since $\mathcal G$ and $\mathcal H$ are reduced, $g^{[\beta]}$ is regular sig-reduced by $\mathcal G$ and $h^{[\gamma]}$ is regular sig-reduced by $\mathcal H$.
But because $\LD(\mathcal G) = \LD(\mathcal H)$, $g^{[\beta]}$ is also regular sig-reduced by $\mathcal H$ (and analogously for $h^{[\gamma]}$ and $\mathcal G$).
With this, Lemma~\ref{lemma:s-reductions} implies that $g$ and $h$ are scalar multiplies of each other,
and we get $g = h$ because both are monic if they are nonzero.

Finally, we show that also $\beta = \gamma$.
To this end, consider $\alpha = \beta - c\cdot \gamma$ with $c \in K$ such that $\sig(\alpha) \prec \sig(\beta)$.
We now distinguish between the two cases: $c = 1$ and $c \neq 1$.
If $c = 1$, then $\alpha$ is a syzygy, as 
$$
    \overline \alpha = \overline \beta - 1  \cdot \overline \gamma = g - g = 0.
$$
Now, if $\alpha \neq 0$, then $\sig(\alpha)$ is a syzygy signature appearing in $\beta$ or $\gamma$, but this would violate the condition that $\mathcal G$ and $\mathcal H$ are syz-reduced.
If $c \neq 1$, then note that $g \neq 0$ as otherwise the first condition of Definition~\ref{def:reduced} would be violated.
Then, $\overline \alpha = g - c\cdot g$ has the same nonzero leading monomial as $g$.
Moreover, $\overline{\alpha}^{[\alpha]}$ is sig-reducible by $\mathcal G$.
But the same element can be used to regular sig-reduce $g^{[\beta]}$, which is a contradiction to the fact that $\mathcal G$ is reduced.
\end{proof}

Just like classical (reduced) Gröbner bases in the free algebra, also (reduced) strong Gröbner bases can be infinite.
However, strong Gröbner bases have the advantage that they can be computed in an incremental fashion by increasing signature,
analogous to how homogeneous Gröbner bases can be computed degree by degree.
Stopping the computation of a strong Gröbner basis at a given signature $\sigma \in \MFr$ yields a strong Gröbner basis up to that signature. 
Under mild assumptions on the bimodule ordering, this resulting basis is guaranteed to be finite.

In what follows, we denote by $\mathcal G_\sigma$,  with fixed $\sigma \in \MFr$, the set of all elements from the reduced strong Gröbner basis $\mathcal G$ of $I\Sig$ with signature smaller than $\sigma$, that is,
$$
    \mathcal G_\sigma = \{ g^{[\beta]} \in \mathcal G \mid \sig(\beta) \prec \sigma\}.
$$
Note that $\mathcal G_\sigma$ forms a strong Gröbner basis up to signature $\sigma$.
More importantly, this set is finite for any choice of $\sigma$ provided that the underlying bimodule ordering is fair.

\begin{lemma}\label{lem:sig-gb-is-finite}
If the bimodule ordering is fair, then $\mathcal G_\sigma$ is finite for any $\sigma \in \MFr$.
\end{lemma}

\begin{proof}
Lemma~\ref{lemma:s-reductions} and Definition~\ref{def:reduced} imply that the reduced strong Gröbner basis cannot contain multiple elements with the same signature.
With this, the result follows from the fact that there are only finitely many signatures smaller than $\sigma$.
\end{proof}

To actually compute the finite set $\mathcal G_\sigma$, one can use~\cite[Alg.~2]{Hof23}, which allows to compute 
a strong Gröbner basis of $I\Sig$ up to any given signature $\sigma$ provided that the used bimodule ordering is fair, see~\cite[Thm.~3.3.16]{Hof23} and~\cite[Rem.~3.3.17]{Hof23}.
Note, however, that the output of~\cite[Alg.~2]{Hof23} need not be reduced.
To actually obtain the reduced strong Gröbner basis $\mathcal G_\sigma$, one has to interreduce this output.

\begin{remark}\label{rem:sig-gb}
The algorithm mentioned above is inefficient due to its reliance on costly module arithmetic.
A more efficient way of computing the set $\mathcal G_\sigma$
involves using the approach outlined in~\cite[Sec.~3.3.3]{Hof23}, which uses algorithms for computing \emph{signature Gröbner bases}.
These algorithms work with pairs $(f,\sig(\alpha))$ rather than labeled polynomials
$f^{[\alpha]}$ and reconstruct the full module representation $\alpha$ a posteriori.
This results in a substantial speedup in practical computations.
\end{remark}
 
\section{A Modular Gr\"obner Basis Algorithm in Free Algebras}
\label{sec:modgb}

In this section, we extend the modular Gröbner basis algorithm from
commutative polynomial rings~\cite{A, elbert83, pfister2011} and $G$-algebras~\cite{DEVT20}
to free algebras~$\Q\<X>$.


To this end, we first introduce some notation.
Fix an integer $N \geq 2$. We denote $\Z_N=\Z/N\Z$. 
For $\frac{a}{b}\in\Q$ with $\gcd(a,b)=1$ and $\gcd(b,N)=1$, we set
$$
    \left(  \frac{a}{b}\right)_{N} = (a+N\mathbb{Z})(b+N\mathbb{Z})^{-1} \in\Z_N.
$$
If $f = \sum_{i=1}^d c_i w_i \in \Q\<X>$ is a polynomial such that $N$ is
coprime to the denominator of any coefficient $c_i$ of $f$, then its \emph{reduction
modulo $N$} is the element 
$$
    f_{N} =  \sum_{i=1}^d \left(c_i\right)_N w_i \in \Z_{N}\<X>.
$$
For a finitely generated ideal $I = \< f_1,\dots,f_r> \ideal \Q\<X>$, we let
$$
    I_N = \< (f_1)_N, \dots, (f_r)_N > \ideal \Z_N\<X>,
$$
if $N$ is coprime to the denominator of any coefficient of any $f_i$.

\begin{remark}
In general, the ideal $I_N$ depends on the chosen generators of $I$.
Different sets of generators of $I$ can lead to different ideals $I_N$.
Therefore, in what follows, we assume a fixed set of generators for $I$
and all further notions are to be understood w.r.t.~the chosen set of generators.
\end{remark}


\subsection{Problems with the Classical Approach}
\label{sec:problems-with-classical-approach}

The basic idea of a modular Gröbner basis algorithm in the classical settings of commutative polynomials and G-algebras is as follows:
To compute a rational Gröbner basis of an ideal $I$,
first choose a set $\mathcal P$ of primes and compute a Gröbner basis $\mathcal G_p$ of the reduction $I_p$ for every $p \in \mathcal P$.
Then, reconstruct from these modular images a Gröbner basis $\mathcal G$ of $I$ over the rationals.
In the following, we illustrate that the infinite nature of Gröbner bases in free algebras renders this approach unfeasible in our setting.

In a modular Gröbner basis algorithm, it is crucial that most algebraic information about the ideal $I$ is preserved when transitioning to $I_p$.
Primes for which this is the case are called \emph{lucky}~\cite{A,DEVT20} and only from lucky primes a reconstruction of the original Gröbner basis 
is possible.
The main property of a lucky prime $p$ is that $\LM(\mathcal G) = \LM(\mathcal G_p)$, and an important
result in the classical settings of commutative polynomials and $G$-algebras is that, for a fixed ideal,
only finitely many primes are not lucky, see~\cite[Thm.~5.13]{A} and~\cite[Lem.~16]{DEVT20}.

In the free algebra, however, this is no longer true. 
In $\Q\<X>$, there exist finitely generated ideals for which there are infinitely many primes that are not lucky.
In fact, as the following example shows, there even exist principal ideals for which no prime is lucky regardless of the chosen monomial ordering.

\begin{example}
\label{ex:infinitely-many-unlucky-primes}
   Consider the ideal 
   $$
    I = \langle xyx - xy - y\rangle \ideal \Q\<x,y>.
    $$
    The reduced Gröbner basis of $I$ w.r.t.~any monomial ordering is given by 
    \[
        \mathcal G = \left\{ x y^n x + \frac{F_{n-1}}{F_n} y^n x - \frac{F_{n+1}}{F_n} xy^n - y^n \mathrel{\bigg|} n \geq 1\right\},
    \]
    where $(F_n)_{n \geq 0}$ denotes the Fibnoacci sequence, that is, $F_0 = 0, F_1 = 1$ and $F_{n+1} = F_{n} + F_{n-1}$ for all $n > 1$.
    
    To see this, first note that, regardless of the concrete monomial ordering, we always have $xy^n x \succ y^n x\succ y^n$ and $xy^n x \succ xy^n \succ y^n$, respectively.
    Thus, $\mathcal G$ is interreduced and all elements in $\mathcal G$ are monic.

    In the following, we abbreviate $a_n = \frac{F_{n-1}}{F_n}$ and $b_n = \frac{F_{n+1}}{F_n}$, and we let
    $g_n = xy^nx + a_n y^n x - b_n xy^n - y^n$ for $n \geq 1$.
    Then, we can see that $\<\mathcal G> = I$, because $xyx - xy - y = g_1 \in \mathcal G$ and we get inductively that
    \begin{align}
    \label{eq:recurion-generator}
        g_{n+1} = \frac{-1}{b_n}\bigg( g_n y (x - 1) - (x + a_n) y^n g_1 \bigg) \in I,
    \end{align}
    for all $n \geq 2$.
    A proof of the identity~\eqref{eq:recurion-generator} is given in Appendix~\ref{appendix}.
    
    Finally, to verify that $\mathcal G$ is a Gröbner basis, one can check that all S\nobreakdash-polynomials of $\mathcal G$ reduce to zero.    
    This is a straightforward, if tedious, calculation, which is detailed in Appendix~\ref{appendix}.
    
    Thus, 
    $$
        \<\LM(I)> = \<\LM(\mathcal G)> = \<x y^n x \mid n \geq 1>.
    $$
    
    Now, fix an arbitrary prime $p\in\N$ and let $N$ be the smallest positive integer such that $p \mid F_N$.
    It is a well-known property of the Fibonacci sequence that such an $N$ always exists~\cite{Wall60}.
    We then claim that 
    $$
        y^N x - x y^N \in I_p = \< xyx - xy -y> \ideal \Z_p\<X>,
    $$ 
    which implies that $\LM(I) \neq \LM(I_p)$, and therefore, that $p$ cannot be lucky.

    To verify this claim, we first note that the representation~\eqref{eq:recurion-generator} implies that 
    $(g_n)_p \in I_p$ for all $n \leq N-1$, because $a_n$ and $1/b_n$ are well-defined modulo $p$ for $1 \leq n \leq N-2$.
    Therefore,
    \begin{align*}
        (g_{N-1})_p &= xy^{N-1}x + \underbrace{(a_{N-1})_p}_{=\,-1} y^{N-1} x + \underbrace{(b_{N-1})_p}_{=\,0} xy^{N-1} - y^{N-1}\\
                    &= xy^{N-1}x - y^{N-1} x - y^{N-1} \in I_p.
    \end{align*}
    With this, it is straightforward to verify that 
    $$
        y^N x - x y^N = (g_{N-1})_p y (1-x) + (x-1)y^{N-1} (g_1)_p \in I_p.
    $$
    This shows that $\LM(I) \neq \LM(I_p)$, and thus, that no prime is lucky for $I$ (regardless of the monomial ordering). 
\end{example}

We summarize the conclusion of the example in the following theorem.

\begin{theorem}
For $|X| > 1$, there exist finitely generated ideals in $\Q\<X>$ for which no prime is lucky regardless of the monomial ordering.
\end{theorem}


%

The underlying reason for the phenomenon occurring in Example~\ref{ex:infinitely-many-unlucky-primes}
is that (reduced) Gröbner bases in free algebras can be infinite.
In practice, however, we can only compute finite subsets of such infinite bases anyway. 
This limitation to finite computations helps mitigate the issue of not having any lucky primes but it introduces a new challenge: 
How can we measure progress in an infinite Gröbner basis computation?

Addressing this issue is crucial, as it involves several related considerations. 
In particular, how can we effectively halt the enumeration of (potentially infinite) Gröbner bases ${\mathcal G}_p$ of $I_p$ for different moduli $p$ \emph{at the same stage}?
This question is relevant for the reconstruction process, where we have to identify Gröbner basis elements from different $\mathcal G_p$ corresponding to each other. 
Typically, this matching is done by selecting elements across different $\mathcal G_p$ with the same leading monomial.
But given an element $f \in \mathcal G_p$, how do we determine during the computation of $\mathcal G_q$ ($q \neq p$) when -- or whether at all (in case one of the primes is not lucky) -- a corresponding element $g$ with $\LM(f) = \LM(g)$ will be computed?

Furthermore, assuming that we can resolve the issue of identifying corresponding elements and that we can reconstruct a partial Gröbner basis from the modular images, how can we characterize (and then also verify) that this finite set forms part of an infinite (reduced) Gröbner basis, without knowing the full infinite basis?

For homogeneous ideals, both problems can be handled by using the degree as a measure of computational progress.
First, we compute all partial Gröbner bases up to a fixed degree $D$ and then we can verify that the reconstructed
result is indeed a homogeneous Gröbner basis up to degree $D$. 

In the case of inhomogeneous ideals, however, we are not aware of any such direct measure.
One could, of course, homogenize an inhomogeneous ideal, compute partial homogeneous Gröbner bases, and then dehomogenize the final result. 
However, it has been observed that this process can be highly inefficient in the free algebra, transforming, in the worst case, a finite Gröbner basis computation into an infinite one~\cite{Ufn08}.

In the following section, we describe how signatures and signature-based algorithms can be used to address these issues. 
Specifically, we will use signatures as a measure of computational progress, naturally generalizing the degree used in the homogeneous case.

\subsection{A Signature-based Approach}

For what follows, we fix a family of polynomials $(f_1,\dots,f_r) \in \Q\<X>^r$, which generates an ideal $I = \langle f_1,\dots,f_r \rangle$ and we
let 
$$
    I\Sig = \left\< f_1^{[\varepsilon_1]},\dots,  f_r^{[\varepsilon_r]}\right>
$$ be the labeled module generated by $f_1,\dots,f_r$.

For an integer $N\geq 2$ for which the reductions $(f_1)_N,\dots,(f_r)_N$ are well\nobreakdash-defined,
we let $I_N\Sig$ be the labeled module generated by $(f_1)_N, \dots, (f_r)_N$ over $\Z_N\<X>$, that is,
$$
    I_N\Sig =  \left\< (f_1)_N^{[\varepsilon_1]},\dots,  (f_r)_N^{[\varepsilon_r]}\right>,
$$
generated as a $\Z_N\<X>$-bimodule.

\begin{remark}
\label{rem:well-defined-reduction}
There are only finitely many $N$ for which $(f_1)_N,\dots,(f_r)_N$ are not well-defined,
namely the divisors of the denominators of coefficients in the $f_i$.
For a simpler presentation, we assume in the following that all chosen (prime) numbers $N$ are such that $I_N\Sig$ is well-defined.
\end{remark}

The goal of our modular approach is to reconstruct the reduced strong Gröbner basis of $I\Sig$ up to a fixed signature $\sigma$ from its modular images.
We need to make an assumption on the bimodule ordering in use.

\begin{assumption}
Throughout the rest of this article, we assume that the bimodule ordering on $\MFr$ is fair.
\end{assumption}

With this assumption, our approach provides two main advantages over the classical method:
First, by Lemma~\ref{lem:sig-gb-is-finite}, all involved reduced strong Gröbner basis computations are guaranteed to be finite,
and second, we can leverage the data of signatures to assess computational progress and to identify corresponding Gröbner basis elements across different modular images.
Furthermore, as we will discuss in Section~\ref{sec:verification}, the final verification test in this approach will be both more general and efficient than the classical method.

Just like the classical approach, also our method relies on the assumption that essential algebraic information is 
retained when transitioning from $I\Sig$ to the reduction $I\Sig_p$ modulo a prime $p$.
In particular, the information encoded within the signatures has to be preserved during this transition.
This leads to the concept of \emph{$\sigma$-lucky primes} introduced below, which adapts the definition of lucky primes to our signature-based setting.

But first, we fix some notation.
For a prime $p \in \N$, we denote by $\mathcal G$ and $\mathcal G_p$ the reduced strong Gröbner basis of $I\Sig$ and $I\Sig_p$, respectively.
Recall that, for fixed $\sigma \in \MFr$, we denote by $\mathcal G_\sigma$ (resp.~$\mathcal G_{p,\sigma}$) the set
$$
    \mathcal G_\sigma = \{ g^{[\beta]} \in \mathcal G \mid \sig(\beta) \prec \sigma\}
$$
(resp.~$\mathcal G_{p,\sigma} = \{ g^{[\beta]} \in \mathcal G_p \mid \sig(\beta) \prec \sigma\}$).

We also extend the concept of reduction modulo $N$ from polynomials to module elements.
If $\alpha = \sum_{i=1}^d c_i \mu_i \in \Sigma$ is a module element and $N \geq 2$ is an integer such that $N$ is
coprime to the denominator of any coefficient $c_i$ of $\alpha$, then its \emph{reduction
modulo $N$} is the element 
$$
    \alpha_{N} =  \sum_{i=1}^d \left(c_i\right)_N \mu_i.
$$
Analogously, if $\mathcal F \subseteq I\Sig$ is a set of labeled polynomials such that $\alpha_N$ is well-defined
for all $f^{[\alpha]} \in\mathcal F$, we set
$$
    \mathcal F_N = \{ (f_N)^{[\alpha_N]} \mid f^{[\alpha]} \in \mathcal F\}.
$$
Note that the condition on $\alpha_N$ (in combination with Remark~\ref{rem:well-defined-reduction}) implies that also $f_N$ is well-defined.

We can now define the concept of $\sigma$-lucky primes.

\begin{definition}\label{def:sigma-lucky}
Let $\sigma \in \MFr$. 
A prime $p \in \N$ is \emph{$\sigma$-lucky} if
\begin{enumerate}
    \item $(\mathcal G_\sigma)_p$ is well-defined;\label{cond:lucky1}
    \item $\LD(\mathcal G_\sigma) = \LD(\mathcal G_{p,\sigma})$.\label{cond:lucky2}
\end{enumerate}
Otherwise, $p$ is called \emph{$\sigma$-unlucky}.
\end{definition}

Thus, $\sigma$-lucky primes preserve the information encoded in the leading data of the reduced strong Gröbner basis when working in $I\Sig_p$ instead of $I\Sig$.
This is more strict than the classical definition of lucky primes.
Nevertheless, the set of $\sigma$-unlucky primes is still finite.
This follows from the following, more general, result, which
relates the reduced strong Gröbner basis $\mathcal G_{p,\sigma}$ of $I\Sig_p$
to the reduction $(\mathcal G_\sigma)_p$ 
of the reduced strong Gröbner basis $\mathcal G_\sigma$ of $I\Sig$.

\begin{lemma}\label{lem-ld-prime}
For any prime $p$, if $\LD((\mathcal G_\sigma)_p) = \LD(\mathcal G_\sigma)$, then $(\mathcal G_\sigma)_p$ = $\mathcal G_{p,\sigma}$.
\end{lemma}

\begin{proof}
Let $p$ be such that $\LD((\mathcal G_\sigma)_p) = \LD(\mathcal G_\sigma)$.
By Theorem~\ref{thm:cover-theorem-free-algebra}, which we state and discuss in more detail later,
checking whether a set of labeled polynomials forms a strong Gröbner basis (up to signature $\sigma$)
only relies on information encoded in the leading data of this set.
Consequently, since $\mathcal G_\sigma$ is a strong Gröbner basis of $I\Sig$ up to signature $\sigma$,
the set $(\mathcal G_\sigma)_p$ forms a strong Gröbner basis of $I\Sig_p$ up to signature $\sigma$.
In fact, since $\mathcal G_\sigma$ is reduced, so is $(\mathcal G_\sigma)_p$.
With this, the uniqueness of the reduced strong Gröbner basis implies that 
$(\mathcal G_\sigma)_p$ = $\mathcal G_{p,\sigma}$.
\end{proof}

\begin{corollary}\label{cor finite} 
For any $\sigma \in \MFr$, the set of $\sigma$-unlucky primes is finite.
\end{corollary}

\begin{proof}
By Lemma~\ref{lem:sig-gb-is-finite}, the set $\mathcal G_\sigma$ is finite, 
say $\mathcal G_\sigma = \{ g_1^{[\beta_1]}, \dots, g_s^{[\beta_s]}\}$.
Any prime $p$ which does not divide the denominator of any coefficient appearing in any of the $\beta_i$ satisfies the first condition of Definition~\ref{def:sigma-lucky}.
Among all the primes satisfying the first condition, let $p$ be such that
$\LD((\mathcal G_\sigma)_p) = \LD(\mathcal G_\sigma)$.
Note that there are only finitely many primes which do not qualify here.
Then, Lemma~\ref{lem-ld-prime} yields $(\mathcal G_\sigma)_p$ = $\mathcal G_{p,\sigma}$,
which implies the second condition in Definition~\ref{def:sigma-lucky}.
\end{proof}





If $p$ is $\sigma$-lucky, then $\mathcal G_{p,\sigma}$ and $(\mathcal G_\sigma)_p$ are actually equal.

\begin{corollary}
\label{cor:L2-cond}
If $p$ is $\sigma$-lucky, then $(\mathcal G_\sigma)_p = \mathcal G_{p, \sigma}$.
\end{corollary}

\begin{proof}
We show that $p$ $\sigma$-lucky implies that $\LD((\mathcal G_\sigma)_p) = \LD(\mathcal G_\sigma)$.
With this, Lemma~\ref{lem-ld-prime} yields the claimed result.

Let $g^{[\beta]} \in \mathcal G_\sigma$.
We show that $\LD(g^{[\beta]}) = \LD(g_p^{[\beta_p]})$.
If $g = 0$, then this follows immediately from the requirement that $\beta$ has to be monic in a reduced set (first condition in Definition~\ref{def:reduced}).
Otherwise, that is, if $g \neq 0$, then $g$ is monic and we get $\lm(g) = \lm(g_p)$.
To see that also the signatures have to be the same,
note that $g_p^{[\beta_p]}$ is sig-reducible by $\mathcal G_{p, \sigma}$ as the latter is a strong Gröbner basis of $I\Sig_p$ up to signature $\sigma$.
But because $p$ is $\sigma$-lucky, sig-reducibility of $g_p^{[\beta_p]}$ by $\mathcal G_{p, \sigma}$
would imply regular sig-reducibility of $g^{[\beta]}$ by $\mathcal G_{\sigma}$ if $\sig(\beta) \neq \sig(\beta_p)$.
Since this would violate the fact that $\mathcal G_{\sigma}$ is reduced, we must have $\sig(\beta) = \sig(\beta_p)$ as required.
\end{proof}

We are now in the same situation as in the classical cases:
For all but finitely many primes, the reduction of $\mathcal G_\sigma$ modulo $p$ equals the reduced strong Gröbner basis $\mathcal G_{p,\sigma}$ of $I\Sig_p$.
Thus, we can now follow the classical approach and, for a finite set of primes $\mathcal P$, first compute $\mathcal G_{p,\sigma}$ for all $p \in \mathcal P$
and then reconstruct $\mathcal G_\sigma$ from these modular images.

The reconstruction consists of two steps. 
First, Chinese remaindering is used to lift the sets $\mathcal G_{p,\sigma}$ to a set of labeled polynomials $\mathcal G_{N,\sigma} \subseteq I\Sig_N$, where $N = \prod_{p \in \mathcal P}$.
Here, to identify Gröbner basis elements corresponding to each other, we compare their signatures.
Then, rational reconstruction is used to lift the coefficients occurring in $\mathcal G_{N,\sigma}$ to rational coefficients.
Classically, this reconstruction is done via the Farey rational map, which is guaranteed to be bijective if $\sqrt{N/2}$ is larger than the absolute value of all numerators and denominators in $\mathcal G_\sigma$~\cite{WGD82}, see also~\cite{KG83}.
We note that there are also more sophisticated approaches for rational reconstruction, see, e.g.,~\cite{Mon04,bdfp},
but we will stick to the classical Farey map in the following as this is also the method implemented in \textsc{SageMath}.

Assuming that $\mathcal P$ only consists of $\sigma$-lucky primes and that $\mathcal P$ is sufficiently large, 
this process is guaranteed to reconstruct
the reduced strong Gröbner basis $\mathcal G_\sigma$ of $I\Sig$ up to signature $\sigma$.
The following definition formalizes when $\mathcal P$ is sufficiently large.


\begin{definition}
Let $\sigma \in \MFr$.
A finite set $\mathcal P$ of $\sigma$-lucky primes is \emph{sufficiently large for $\sigma$}
if
$$
    \prod_{p \in \mathcal P} p \geq \max \{ 2 \cdot c^2 \mid c \in \Z \text{ numerator or denominator appearing in } \mathcal G_\sigma\}.
$$
\end{definition}

\begin{proposition}
\label{prop:lifting-correct}
Let $\sigma \in \MFr$.
If $\mathcal{P}$ is a set of $\sigma$-lucky primes which is sufficiently large for $\sigma$,
then the reduced strong Gr\"obner bases $\mathcal G_{p,\sigma}$ up to signature $\sigma$ ($p\in\mathcal{P}$) lift via Chinese remaindering
and rational reconstruction to the reduced strong Gr\"obner basis $\mathcal G_\sigma$ up to signature $\sigma$. 
\end{proposition}
\begin{proof} 
Since all primes in $\mathcal{P}$ are $\sigma$-lucky, Corollary~\ref{cor:L2-cond} implies that $(\mathcal G_\sigma)_p = \mathcal G_{p,\sigma}$ for all $p\in\mathcal{P}$.
Since $\mathcal{P}$ is assumed to be sufficiently large for $\sigma$, the coefficients of the Chinese
remainder lift $\mathcal G_{N,\sigma}$, with $N = \prod_{p \in \mathcal P}$, have a unique preimage under the Farey map, that is, 
they have a unique lift to the rational numbers.
\end{proof}

Corollary~\ref{cor finite} guarantees the existence of a sufficiently large set $\mathcal P$ of $\sigma$-lucky primes.
However, in practice, we face the problem that we cannot decide in advance whether a given set $\mathcal P$ is indeed sufficiently large,
or if it consists only of $\sigma$-lucky primes.
Both these properties would require knowledge of $\mathcal G_\sigma$, which is obviously unknown a priori.

To remedy this situation, we proceed as suggested in~\cite{bdfp,DEVT20}:
We randomly choose a finite set $\mathcal P$ of primes and compute, for fixed $\sigma \in \MFr$, the set $\mathcal{GP} = \{ \mathcal G_{p,\sigma} \mid p \in \mathcal P\}$.
Then, we use a majority vote to filter out $\sigma$-unlucky primes with high probability:

\vspace{0.2cm}
\noindent\emph{\textsc{deleteByMajorityVote:} 
Define an equivalence relation on $(\mathcal{P}, \mathcal{GP})$ by setting $(p, \mathcal G_{p,\sigma}) \sim (q, \mathcal G_{q,\sigma}):\Longleftrightarrow \LD(\mathcal G_{p,\sigma})= \LD(\mathcal G_{q,\sigma})$.
Then, replace $(\mathcal{P}, \mathcal{GP})$ by an equivalence class of largest cardinality.}
\vspace{0.2cm}

After this deletion step, the set $\mathcal P$ only contains $\sigma$-lucky primes with high probability.
We note that a faulty decision here will be detected by subsequent tests.
Moreover, after the majority vote, all reduced strong Gröbner bases in $\mathcal{GP}$ share the same leading data.
Hence, we can apply Chinese remaindering and rational reconstruction to the coefficients appearing in these bases.
If the reconstruction fails (because $\mathcal P$ was not sufficiently large), 
we enlarge the set $\mathcal{P}$ by new primes not used so far, and repeat the whole process.

\begin{remark}
Although highly unlikely in practice, it could happen that 
\emph{\textsc{deleteByMajorityVote}} only accumulates $\sigma$-unlucky primes
if done naively. 
To avoid this, a weighted cardinality count must be employed.
Whenever the set $\mathcal P$ is enlarged, the newly added primes should be given
a collective weight that exceeds the total weight of the previously present elements.
\end{remark}

Once the reconstruction succeeds and we obtain a set $\mathcal G_\sigma$ of labeled polynomials with rational coefficients,
we conduct a final verification test to ensure that the reconstructed result is indeed a strong Gröbner basis of $I\Sig$ up to signature $\sigma$.
The details of this verification process will be discussed further in Section~\ref{sec:verification}; 
for now, we treat it as a black-box test.

All in all, we obtain the following Algorithm~\ref{modGalgo} for computing strong Gröbner bases in a modular fashion.
It closely resembles~\cite[Alg.~6]{bdfp} or~\cite[Alg.~1]{DEVT20}.

\begin{algorithm}
\caption{Modular Gr\"obner Basis Algorithm}\label{modGalgo}
\begin{algorithmic}[1]
\REQUIRE A family of polynomials $(f_1,\dots,f_r) \in \Q\<X>^r$ generating an ideal $I = \< f_1,\dots,f_r>$ and a signature $\sigma \in \MFr$.
\ENSURE The reduced strong Gröbner basis $\mathcal G_\sigma$ of the labeled module $I\Sig$ up to signature $\sigma$.
\STATE choose randomly a set $\mathcal P$ of primes;
\WHILE{\textsc{True}}
\STATE compute $\mathcal{GP} = \{\mathcal G_{p,\sigma} \mid p \in \mathcal P\}$;
\STATE $(\mathcal{P},\mathcal{GP})=\text{\emph{\textsc{deleteByMajorityVote}}}(\mathcal{P},\mathcal{GP})$;
\STATE lift the strong Gr\"obner bases in $\mathcal{GP}$ to $\mathcal G_\sigma\subseteq I\Sig$ via Chinese remaindering and rational reconstruction;
\IF {the lifting succeeds and \emph{\textsc{verificationTest}}$(\mathcal G_\sigma, \sigma)$}
\RETURN $\mathcal{G}_\sigma$;
\ENDIF
\STATE enlarge $\mathcal{P}$ with primes not used so far;
\ENDWHILE
\end{algorithmic}
\end{algorithm}

\begin{remark}
Algorithm~\ref{modGalgo} offers several opportunities for parallel computation. 
Most notably, the computation of $\mathcal G_{p,\sigma}$ can be parallelized for different primes $p$.
Additionally, the final verification test can also be performed in parallel, see Remark~\ref{rem:final-test-parallel}.
\end{remark}

\begin{remark}
Algorithm~\ref{modGalgo} can be adapted to handle one-sided ideals, allowing the computation of one-sided strong Gröbner bases. 
This can be done in one of two ways: either by directly substituting all two-sided Gröbner basis computations with their one-sided counterparts, or by applying the techniques from~\cite{Hey01}, which convert one-sided Gröbner basis computations into two-sided ones, allowing the existing algorithm to be used without modification.
\end{remark}

\subsection{Verification}
\label{sec:verification}

Proposition~\ref{prop:lifting-correct} ensures that the reduced strong Gröbner basis $\mathcal G_\sigma$ up to signature $\sigma$ can be recovered from modular computations provided that the set $\mathcal P$ of used primes is sufficiently large for $\sigma$.
Since we cannot verify in advance whether $\mathcal P$ is sufficiently large,
we have to conduct a final verification test to ensure the correctness of Algorithm~\ref{modGalgo}.

As mentioned earlier, a verification of classical Gröbner bases in free algebras relies on an analogue of Buchberger's S-polynomial criterion.
A difference to the commutative case is, however, that two noncommutative polynomials can give rise to more than one S-polynomial.
To uniquely characterize each S-polynomial, we recall the concept of \emph{ambiguities} from~\cite{Ber78}.
We adapt this notion directly to labeled polynomials and note that the definition for classical
polynomials can be recovered by simply ignoring the bimodule labelings.

Let $f^{[\alpha]}, g^{[\beta]} \in I\Sig$ with $f,g \neq 0$.
If $a,b,c,d \in \<X>$ are such that $\lm(afb) = \lm(cgd)$ and 
$a,b,c,d$ are minimal in a certain sense, then the tuple
$$
    \amba = (a \otimes b, c \otimes d, f^{[\alpha]}, g^{[\beta]})
$$
is called an \emph{ambiguity} of $f^{[\alpha]}$ and $g^{[\beta]}$.
For simplicity, omit the precise characterization of minimality here, as it is not relevant for our purpose.
A full definition can be found in, e.g.,~\cite[Def.~2.4.46]{Hof23}.

Each ambiguity gives rise to an S-polynomial.
To verify that a set of polynomials forms a Gröbner basis, one must check whether all S-polynomials reduce to zero.
This is computationally expensive as it involves extensive polynomial arithmetic.
In contrast, verifying that a set is a \emph{strong} Gröbner basis is a purely combinatorial task and does not require any polynomial arithmetic.
This is the result of the \emph{cover criterion} stated below, which is based on the concept of \emph{covered} ambiguities.

In the following, for an ambiguity $\amba = (a \otimes b, c \otimes d, f^{[\alpha]}, g^{[\beta]})$, we define
$$
\lma(\amba) = \lm(afb)
\quad
\text{ and }
\quad
\siga(\amba) = \max\{ \sig(a\alpha b), \sig(c \beta d)\}.
$$
Moreover, $\amba$ is \emph{regular} if $\sig(a\alpha b) \neq \sig(c \beta d)$.

\begin{definition}\label{def:cover-free-algebra}
Let $\mathcal G \subseteq I\Sig$ and $f^{[\alpha]}, g^{[\beta]} \in \mathcal G$ with $f,g \neq 0$.
An ambiguity $\amba$ of $f^{[\alpha]}$ and $g^{[\beta]}$ is \emph{covered} (by $\mathcal G$) 
if there exist $h^{[\gamma]} \in \mathcal G$ and $a,b\in \<X>$ such that the following conditions hold:
\begin{enumerate}
	\item $\lma(\amba) \succ \lm(a h b)$;
	\item $\siga(\amba) = \sig(a \gamma b)$.
\end{enumerate}
\end{definition}

In the definition, $h$ can also be zero, in which case the first condition is always fulfilled since $\lm(0) = 0 \prec w$ for all $w \in \<X>$.

The cover criterion characterizes strong Gröbner bases via covered ambiguities.
It was originally formulated for commutative polynomials~\cite[Thm.~2.4]{GVW16} and has 
subsequently been generalized to the free algebra~\cite[Thm.~3.3.11]{Hof23}.

\begin{theorem}[Cover criterion]
\label{thm:cover-theorem-free-algebra}
Let $\sigma \in \MFr$.
A subset $\mathcal G \subseteq I\Sig$ is a strong Gröbner basis of $I\Sig$ (up to signature $\sigma$) if the following conditions hold:
\begin{enumerate}
	\item for all $i = 1,\dots,r$ (with $\varepsilon_i \prec \sigma$), there exists $g_i^{[\beta_i]} \in \mathcal G$ with $\sig(\beta_i) =~\varepsilon_i$;
    \item for all $f^{[\alpha]}, g^{[\beta]} \in \mathcal G$ and $w \in \<X>$ (with $\sig(\alpha w g - f w \beta) \prec \sigma)$, the syzygy $0^{[\alpha w g - f w \beta]}$ is top syz-reducible by $\mathcal G$;
    \item all regular ambiguities $\amba$ of $\mathcal G$ (with $\siga(\amba) \prec \sigma$) are covered.
\end{enumerate}
\end{theorem}

Based on the cover criterion, we can formulate the following final verification test for Algorithm~\ref{modGalgo}.
An implicit assumption in the cover criterion is that all elements $g^{[\beta]}$ are well-formed labeled polynomials, meaning
they satisfy $\overline \beta = g$.
For the output of our modular algorithm, we cannot assume that this condition holds automatically.
Therefore, we must explicitly verify it.

\vspace{0.2cm}
\emph{\textsc{verificationTest:} 
First, check for every element $g^{[\beta]}$ in the reconstructed candidate basis that indeed $\overline \beta = g$.
Then, check the cover criterion up to signature $\sigma$.
Return \texttt{true} if both tests pass, and \texttt{false} otherwise.}
\vspace{0.2cm}

It is important to note that this verification test applies to all inputs, including inhomogeneous ones. 
This is in contrast to the classical approach used in other settings, where a rigorous verification is only possible for homogeneous inputs~\cite{DEVT20}. 

\begin{remark}\label{rem:poly-arith}
Verifying the cover criterion is a purely combinatorial problem and can be done very efficiently.  
The more costly part of the verification test is the polynomial arithmetic, i.e., that $\overline \beta = g$ holds for all $g^{[\beta]}$.
However, note that the polynomial arithmetic required for our verification test only grows linearly with the size of the candidate basis.
In contrast, verifying Buchberger's S-polynomial criterion would require to reduce quadratically many S-polynomials to zero, and thus,
is, in general, more expensive.
\end{remark}

As mentioned in Remark~\ref{rem:sig-gb}, signature-based algorithms typically first compute a signature Gröbner basis,
which consists of pairs $(g,\sig(\beta))$ rather than labeled polynomials $g^{[\beta|}$.
Afterwards, the full module representations $\beta$ are reconstructed. 
This approach is faster than directly computing a strong Gröbner basis, although the reconstruction step may still be time-consuming.

We can observe that the information provided by a signature Gröbner basis is already sufficient for verifying the cover criterion,
as the latter only relies on information contained in the leading data.
Therefore, if the goal is solely to compute a Gröbner basis without requiring the module representations, there is no need to reconstruct them in full.

Using this observation, we can adapt Algorithm~\ref{modGalgo} as to only compute the \emph{reduced signature Gröbner basis} of $I\Sig$ (up to signature $\sigma$).

\begin{definition}
Let $\mathcal G \subseteq I\Sig$ be the reduced strong Gröbner basis (up to signature $\sigma \in \MFr$).
The \emph{reduced signature Gröbner basis} of $I\Sig$ (up to signature $\sigma$) is the set
$$
\{ (g, \sig(\beta) \mid g^{[\beta]} \in \mathcal G\} \subseteq I \times \MFr.
$$
\end{definition}

The reduced signature Gröbner basis inherits most of the properties from the reduced strong Gröbner basis.
In particular, it is unique and the nonzero polynomials form a Gröbner basis of the underlying ideal.

We can replace all reduced strong Gröbner basis computations in Algorithm~\ref{modGalgo} by reduced signature Gröbner basis computations.
Then, for the final verification, we have to verify the cover criterion and additionally 
we must also ensure that, for every pair $(g, \mu) \in \Q\<X> \times \MFr$ that we reconstruct from modular images, 
$g$ and $\mu$ are correctly related.
This means we need to verify the existence of $\beta \in \Sigma$ such that $\overline{\beta} = g$ and $\sig(\beta) = \mu$.

The following result shows how this verification can be done efficiently.
It holds for arbitrary coefficient fields $K$ and is therefore phrased in this general setting.
We also note that the concept of regular sig-reductions can be extended from labeled polynomials to pairs in $K\<X> \times \MFr$.
In particular, $(f,\sigma)$ is regular sig-reducible by $(g,\mu)$ if there are $a,b \in\<X>$ such that $\lm(agb) \in \supp(f)$ and $a\mu b \prec \sigma$.
The corresponding regular remainder is the pair $(f - agb, \sigma)$.

\begin{proposition}
Let $\mathcal H \subseteq K\<X> \times \MFr$ with every nonzero polynomial in $\mathcal H$ monic
and such that 
$$(\lm(g), \mu) = (\lm(ahb), a\nu b) \implies (g,\mu) = (h, \nu),
$$
for all $(g,\mu), (h, \nu) \in \mathcal H$ and $a,b \in \<X>$.

Then $\mathcal H$ is the reduced signature Gröbner basis of $I\Sig$ (up to signature $\sigma \in \MFr$) if and only if the following conditions hold:
\begin{enumerate}
    \item $\mathcal H$ satisfies the cover criterion (up to signature $\sigma$);
    \item for every $(g,\mu) \in \mathcal H$, the element $(\overline{\mu}, \mu)$ regular sig-reduces by $\mathcal H$ to $(c\cdot g, \mu)$ for some nonzero $c \in K$.
\end{enumerate}
\end{proposition}

\begin{proof}
For the first implication, assume that $\mathcal H$ is the reduced signature Gröbner basis of $I\Sig$ (up to signature $\sigma$)
and let $\mathcal G$ be the reduced strong Gröbner basis (up to signature $\sigma$).
Then $\mathcal H$ clearly satisfies the cover criterion (up to signature $\sigma$) because $\mathcal G$ does.
For the second condition, let $(g,\mu) \in \mathcal H$.
By assumption, there is $\beta \in \Sigma$ such that $g^{[\beta]} \in \mathcal G$ and $\sig(\beta) = \mu$.
Then, let $(h,\mu)$ be the result of regular sig-reducing $(\overline{\mu},\mu)$ by $\mathcal H$ and let $\gamma \in \Sigma$ be such that $\overline{\gamma} = h$ and $\sig(\gamma) = \mu$.
The labeled polynomial $h^{[\gamma]}$ is regular sig-reduced by $\mathcal G$.
Furthermore, so is $g^{[\beta]}$, because $\mathcal G$ is reduced.
Since $\sig(\beta) = \mu = \sig(\gamma)$, Lemma~\ref{lemma:s-reductions} yields that $h = c\cdot g$ for some nonzero $c \in K$.

For the other implication, assume that $\mathcal H$ satisfies the two conditions of the proposition.
First, we show that for every $(g,\mu) \in \mathcal H$, there exists $\beta \in \Sigma$ 
such that $\overline{\beta} = g$ and $\sig(\beta) = \mu$.
Assume, for contradiction, that this is not the case and let $(g,\mu) \in \mathcal H$ be with $\mu$ minimal such that
there does not exist such a $\beta$.
Denote $\mathcal H_{\mu}= \{ h^{[\beta|} \in I\Sig \mid (h,\sig(\beta)) \in \mathcal H, \sig(\beta) \prec \mu\}$.
Then, regular sig-reducing $(\overline{\mu}, \mu)$ by $\mathcal H$ yields the same polynomial part as 
regular sig-reducing the labeled polynomial $\overline{\mu}^{[\mu]}$ by $\mathcal H_{\mu}$.
Thus, by the second condition of the proposition, $\overline{\mu}^{[\mu]}$ regular sig-reduces by $\mathcal H_\mu$ to $(c \cdot g)^{[\gamma]}$ 
for some $\gamma \in \Sigma$ with $\sig(\gamma) = \mu$.
Then $\beta = 1 / c \cdot \gamma$ satisfies $\overline{\beta} = g$ and $\sig(\beta) = \mu$.

So, we can pick, for each $(g,\mu)\in \mathcal H$, a $\beta \in \Sigma$ such that $\overline{\beta} = g$ and $\sig(\beta) = \mu$.
In particular, we can pick $\beta$ to be monic in case $g = 0$.
Moreover, by picking the $\beta$'s by increasing signature, we can assume that the labeled polynomial $g^{[\beta]}$ is not syz-reducible by all previously constructed elements.
Collecting these labeled polynomials for all $(g,\mu) \in \mathcal H$
in a set $\mathcal G$, we obtain the reduced strong Gröbner basis of $I\Sig$ (up to signature $\sigma$).  

To see this, note that the first condition of the proposition implies that $\mathcal G$ is a strong Gröbner basis (up to signature $\sigma$).
Moreover, by assumption, all nonzero polynomials are monic and, by construction, so are the module representations of all syzygies.
The second condition of the proposition implies that all elements in $\mathcal G$ are regular sig-reduced
and the assumption on the leading data of $\mathcal H$ implies that they are also top sig-reduced.
Finally, the $\beta$'s were chosen so that each element $g^{[\beta]} \in \mathcal G$ 
is syz-reduced, showing that $\mathcal G$ is reduced.
\end{proof}

This result leads to the following final verification test to check whether the output of Algorithm~\ref{modGalgo}
is the reduced signature Gröbner basis up to signature $\sigma$.

\vspace{0.2cm}
\emph{\textsc{sigVerificationTest:}
First, check that the reconstructed candidate basis $\mathcal H$ satisfies 
$(\lm(g), \mu) = (\lm(ahb), a\nu b) \implies (g,\mu) = (h, \nu)$ for all $(g,\mu), (h, \nu) \in \mathcal H$, $a,b \in \<X>$.
Then, check for every element $(g, \mu) \in \mathcal H$ that $(\overline{\mu},\mu)$ regular sig-reduces by $\mathcal H$ to $(c\cdot g,\mu)$ for some nonzero $c \in \Q$.
Finally, check the cover criterion up to signature $\sigma$.
Return \texttt{true} if all three tests pass, and \texttt{false} otherwise.}
\vspace{0.2cm}

\begin{remark}
Analogous to the previous verification test (see Remark~\ref{rem:poly-arith}),
also in this test, the required polynomial arithmetic grows only linearly with the size of the candidate basis.
\end{remark}

\begin{remark}
\label{rem:final-test-parallel}
Both \emph{\textsc{verificationTest}} and \emph{\textsc{sigVerificationTest}} can be parallelized by performing the
polynomial computations as well as the verification of the cover criterion in parallel.
Moreover, in both verification tests, the polynomial computations can be done in positive characteristic, modulo a prime
that has not been used during the computation of the candidate basis, rather than over $\Q$.
This speeds up the verification process, however it comes at the cost of obtaining the correct result only with high probability.  
\end{remark}

\section{Experiments}\label{sec:experiments}

We have implemented our modular algorithm for computing Gr\"obner bases
in free algebras over $\mathbb{Q}$ in the \textsc{SageMath} package \texttt{signature\_gb}~\footnote{available at \url{https://github.com/ClemensHofstadler/signature_gb}},
which provides, to the best of our knowledge, the first and, thus far, only implementation of noncommutative signature-based algorithms.
As described in~\cite[Sec.~6.3]{Hof23}, this package offers
a signature-based variant of the F4 algorithm~\cite{Fau99} for computing
signature and strong Gröbner bases in free algebras.
In particular, we have implemented the version of Algorithm~\ref{modGalgo} described at the end of Section~\ref{sec:verification} for computing 
the reduced signature Gröbner basis up to signature $\sigma$, using \emph{\textsc{sigVerificationTest}} as the final verification test.
Our implementation also allows to perform the modular Gröbner basis computations in parallel, but not the final verification test.

In Table~\ref{table:comparison}, we compare the performance of this algorithm (denoted by ``modular \texttt{sigGB}'')
using the functionality provided by \texttt{signature\_gb} as a back-end for the modular Gröbner
basis computations, to using \texttt{signature\_gb} directly over the rationals (denoted by ``\texttt{sigGB} over $\Q$'').

Our benchmark examples include the following homogeneous systems, which we have gathered from the literature.
We refer to the respective sources for more in-depth discussions on the construction of these examples.

\begin{enumerate}
    \item The example \texttt{metab3} is taken from~\cite[Ex.~5.2]{LL09}
    and consists of the generators of the ideal defining the universal enveloping algebra of the free metabelian Lie algebra in dimension $3$.
    Explicitly, we consider as generators all commutators $[[x_i, x_j], [x_k, x_l]]$ for $1 \leq i,j,k,l \leq 3$
    in $\Q\<x_1,x_2,x_3>$.
    \item The example \texttt{P6} is taken from~\cite{Kel97} 
    and consists of the two homogeneous generators 
    $f_1 = ccc + 2ccb + 3cca + 5bcc + 7aca$,
    $f_2 = bcc + 11bab + 13aaa$
    in $\Q\<a,b,c>$.
    \item The example \texttt{lp1} is taken from~\cite{LL09}
    and consists of the three homogeneous generators
    $
    f_1 = z^4 + yxyx - xy^2x - 3zyxz, 
    f_2 = x^3 + yxy - xyx, 
    f_3 = zyx - xyz + zxz$
    in $\Q\<x,y,z>$.
\end{enumerate}

The number of noncommutative, inhomogeneous benchmark examples in the literature is rather limited.
This is owed mostly to the difficulty of finding inhomogeneous ideals for which a Gröbner basis computation
is challenging but still finite (which is required for comparability).
Since, in our approach, we can use the signature data as a termination measure, 
we do not have to worry about possibly infinite computations.
To obtain challenging inhomogeneous examples, we reuse well-established commutative benchmarks that we simply interpret as noncommutative polynomials.
All inhomogenoeus systems are part of the \textsc{SymbolicData} project~\cite{SymData} and taken from there.

For all examples, a degree-lexicographic monomial ordering is used in combination with the fair bimodule ordering $\prec_\textup{DoPoT}$ for the signature-based computations.
Since the used bimodule ordering is compatible with the signature-degree, we can compute signature Gröbner bases up to certain degree bounds.
The designated degree bounds are indicated by the number after the last ``\texttt{-}'' in the name of each example in Table~\ref{table:comparison}.
So, for example \texttt{lp1-16} means that we compute a signature Gr\"obner basis up to signature $\sigma$, where $\sigma$ is the smallest module
monomial with $\deg(\sigma) = 16$.

All experiments were run with a timeout of 48 hours on a cluster of dual-socket AMD EPYC 7313 @ 3.7GHz machines running Ubuntu 22.04 with a 250 GB memory limit per computation. 
In Table~\ref{table:comparison}, timings are given in the format hh:mm:ss.
Moreover, we abbreviate threads as ``thr'' and indicate a timeout by ``$> 48$h''.

In the first columns of Table~\ref{table:comparison}, we provide some information 
on the size of the different benchmark examples. 
In particular, we list the number of variables (\#vars), the number of generators (\#gens), and their maximal degree (deg).

\begin{table}
  \footnotesize
  \centering
\begin{tabular}{
    l@{\hskip 4pt}
    c@{\hskip 4pt}
    c@{\hskip 4pt}
    c@{\hskip 4pt}
    |
    c@{\hskip 10pt}
    c@{\hskip 5pt}
    c@{\hskip 5pt}
    c@{\hskip 5pt}S} 
 \toprule
\multirow{2}{*}{Example} & \multirow{2}{*}{\#vars} &\multirow{2}{*}{\#gens} & \multirow{2}{*}{deg} & \multirow{2}{*}{\texttt{sigGB} over $\Q$} & & {modular \texttt{sigGB}} & \\
\cline{6-8}
& & & & & {1 thr} & {4 thr} & {8 thr}\\
 \midrule
 \texttt{metab3-16} & 3 & 6 & 4 & 00:12:29 & 00:10:29 & 00:02:51 & 00:01:21 \\
\texttt{P6-12} & 3 & 2 & 3 & 00:03:16 & 00:03:01 & 00:01:31 & 00:01:23 \\
\texttt{P6-13} & & & & 00:46:22 & 00:25:42 & 00:13:37 & 00:09:31 \\
\texttt{lp1-16} & 3 & 3 & 4 & 04:11:30 & 03:05:34 & 00:52:01 & 00:30:55 \\
\texttt{lp1-17} &  &  &  & $> 48$h & 16:44:47 & 05:28:07 & 03:14:43 \\
 \hline
  \texttt{amrhein-10} & 6 & 6 & 2 & 00:15:39 & 00:07:43 & 00:03:22 & 00:02:38 \\
    \texttt{amrhein-11} & & & & 08:13:37 & 02:15:03 & 00:36:39 & 00:28:07 \\
\texttt{cyclic5-11} & 5 & 5 & 5 & 01:10:13 & 00:34:16 & 00:12:03 & 00:07:33 \\
\texttt{cyclic5-12} & & & & 15:12:07 & 03:36:46 & 01:20:40 & 00:58:45 \\
  \texttt{cyclic6-10} & 6 & 6 & 6 & 06:42:06 & 01:21:56 & 00:25:14 & 00:19:56 \\
    \texttt{cyclic6-11} & & & & $> 48$h & 14:50:14 & 05:44:56 & 03:22:30 \\
  \texttt{eco6-14} & 6 & 6 & 3 & 01:51:25 & 02:01:19 & 00:48:50 & 00:26:55 \\
   \texttt{becker-niermann-30} & 3 & 3 & 5 & 12:32:01 & 01:45:27 & 00:44:48 & 00:32:38 \\
      \texttt{becker-niermann-34} & & & & $> 48$h & 13:52:41 & 06:39:05 & 05:40:20 \\
    \texttt{vermeer-17} & 5 & 4 & 5 & 05:05:42 & 03:32:38 & 01:29:02 & 01:05:24 \\
     \bottomrule
\end{tabular}
\caption{Benchmark timings given in the format hh:mm:ss}
\label{table:comparison}
\end{table}

As Table~\ref{table:comparison} shows, the parallelized version of our modular algorithm outperforms the direct application of \texttt{signature\_gb} over the rationals on all benchmark examples.
Even on smaller examples like \texttt{P6-12}, where the computation time is just a few minutes, the modular approach still performs better.
As the complexity of the computations increases, the advantage of the modular algorithm becomes increasingly significant, achieving speedups of over $20\times$ for examples such as \texttt{cyclic6-10} and \texttt{becker-niermann-30} with 8 threads.

With the single-threaded version, the speedups are less dramatic but still significant, reaching up to $7\times$ for \texttt{becker-niermann-30}. 
Also this version outperforms the classical approach over $\Q$ on all benchmark examples, except \texttt{eco6-14}, where it performs slightly worse.

For most benchmark examples, the final verification step, \emph{\textsc{sigVerificationTest}}, accounts for less than 15\% of the total computation time in the single-threaded modular algorithm. 
The only outlier is the \texttt{becker-niermann} family, where the verification takes up 32\% ($\sim 4.5$ hours) of the total computation time for degree 34. 
Using the probabilistic version of the verification test described in Remark~\ref{rem:final-test-parallel}, significantly reduces this cost to less than one percent ($\sim 2.4$ minutes). 
For all benchmark examples, checking the cover criterion takes less than one second and the majority of the time is spent on verifying the correct correspondence between the polynomials and their signatures.

We have to note that our modular algorithm currently appears to be quite memory-inefficient. While running the benchmarks with 16 threads occasionally led to nice speedups compared to using 8 threads -- for instance, \texttt{lp1-16} completed in just 15 minutes with 16 threads -- in many cases, the computation reached the memory limit and had to be aborted. 
Improving the memory efficiency of our implementation is a priority for future work.

\begin{remark}
\label{rem:number-primes}
The efficiency of our modular algorithm depends, in particular, on the number
of Gr\"obner basis computations in positive characteristic before the reconstruction and verification steps. 
In our implementation, we start with the number of available threads and increment, if necessary, by this number.
\end{remark}

%

\section*{Acknowledgements}

V.~L. thanks Wolfram Decker and Hans Sch\"onemann 
 (Kaiserslautern) for discussions.
C.~H. was supported by the LIT AI Lab funded by the State of Upper Austria. A part of this work has been done while both authors were working at the University of Kassel, Germany.

\bibliographystyle{abbrv}

\appendix
\section{Proofs of Example~\ref{ex:infinitely-many-unlucky-primes}}
\label{appendix}

In the following, we provide the remaining proofs of Example~\ref{ex:infinitely-many-unlucky-primes}.
But before we do this, we first collect some properties of the Fibonacci sequence that we will need.

\begin{lemma}
\label{lemma:fib-properties}
Let $(F_n)_{n \geq 0}$ denote the Fibnoacci sequence, i.e., $F_0 = 0, F_1 = 1$ and $F_{n+1} = F_{n} + F_{n-1}$ for all $n > 1$.
For $n \geq 1$, define $a_n = \frac{F_{n-1}}{F_n}$ and $b_n = \frac{F_{n+1}}{F_n}$.
Then the following identities hold for all $m,n \geq 1$:
\begin{multicols}{2}
\begin{enumerate}
    \item $a_{n+1}b_n = 1$;\label{item:fib-1}
    \item $a_n + 1 = b_n$;\label{item:fib-2}
    \item $b_n + 1 = b_n b_{n+1}$;\label{item:fib-3}
     \item $a_m + b_n  = a_n + b_m$; \label{item:fib-4}
    \item $a_m a_n + 1 = (a_n + b_m) a_{m+n}$;\label{item:fib-5}
     \item $b_m b_n + 1 = (a_n + b_m) b_{m+n}$.\label{item:fib-6}
\end{enumerate}
\end{multicols}
\end{lemma}

\begin{proof}\mbox{}

\begin{enumerate}
    \item Follows immediately from the definition.
    \item $a_n + 1 = \frac{F_{n-1} + F_n}{F_n} = \frac{F_{n+1}}{F_n} = b_n$.
    \item Follows from Properties~\eqref{item:fib-1} and~\eqref{item:fib-2}.   
    \item Follows from $a_n + b_m = \frac{F_{m+n}}{F_m F_n}$, which, in turn, follows from the well-known identity $F_{n-1} F_m + F_{m+1} F_n = F_{m+n}$.
    \item Using $a_n + b_m = \frac{F_{m+n}}{F_m F_n}$, we compute
        \begin{align*}
        a_m a_n + 1 &= \frac{F_{m-1}F_{n-1} + F_m F_n}{F_m F_n} = \frac{F_{m+n-1}}{F_m F_n} \\
         &= \frac{F_{m+n}}{F_m F_n} \frac{F_{m+n-1}}{F_{m+n}} = (a_n + b_m) a_{m+n}.
    \end{align*}
    \item Analogously to~\eqref{item:fib-5}, we compute
    \begin{align*}
        b_m b_n + 1 &= \frac{F_{m+1}F_{n+1} + F_m F_n}{F_m F_n} = \frac{F_{m+n+1}}{F_m F_n} \\
         &= \frac{F_{m+n}}{F_m F_n} \frac{F_{m+n+1}}{F_{m+n}} = (a_n + b_m) b_{m+n}. \qedhere
    \end{align*}
\end{enumerate}
\end{proof}

\begin{example}[continues=ex:infinitely-many-unlucky-primes]
We considered the ideal
$$
    I = \langle xyx - xy - y\rangle \ideal \Q\<x,y>
$$
and claimed that the reduced Gröbner basis of $I$ w.r.t.~any monomial ordering is given by 
    \[
        \mathcal G = \left\{ x y^n x + a_n y^n x - b_n xy^n - y^n \mathrel{\bigg|} n \geq 1\right\},
    \]
with $a_n, b_n$ as defined in Lemma~\ref{lemma:fib-properties}.

There remain two outstanding points to complete the proof of this claim. 
The first is verifying the recursive formula~\eqref{eq:recurion-generator}, given by
\begin{align*}
        g_{n+1} = \frac{-1}{b_n}\bigg( g_n y (x - 1) - (x + a_n) y^n g_1 \bigg),
\end{align*}
where $g_n = xy^nx + a_n y^n x - b_n xy^n - y^n$ for $n \geq 1$.

To verify this formula, we partly factor $g_1$ and $g_n$ as follows:
\begin{align*}
    g_1 &= xy(x-1) - y,\\
    g_n &= (x+a_n)y^nx - (b_nx + 1)y^n.
\end{align*}
Plugging these partial factorizations into the right-hand side of~\eqref{eq:recurion-generator} yields
\begin{align*}
    \frac{-1}{b_n}\bigg(- (b_nx + 1) y^{n+1} (x-1) + (x + a_n) y^{n+1} \bigg),
\end{align*}
where we have already cancelled the common term $(x + a_n)y^nxy(x-1)$.
After expanding this polynomial and collecting common terms, we get
\begin{align*}
    x y^{n+1} x + \frac{1}{b_n} y^{n+1} x - \frac{1}{b_n}(b_n + 1)xy^{n+1} - \frac{1}{b_n}(a_n + 1)y^{n+1}.
\end{align*}
Now, the first three items of Lemma~\ref{lemma:fib-properties} imply that the above polynomial equals $g_{n+1}$.

The second remaining claim to prove is that all S-polynomials of $\mathcal G$ reduce to zero.
We note that for each $m,n \geq 1$, the pair $(g_m,g_n) = ( x y^m x + \ldots,
x y^n x + \ldots) \in \mathcal{G}^2$ gives rise to exactly two overlaps and thus to exactly two S-polynomials:
When leading terms overlap as $xy^n\cdot x \cdot y^m x$, the S\nobreakdash-polynomial is $g_n\cdot y^m x - xy^n \cdot g_m$.
The other overlap $xy^m \cdot x \cdot y^nx$ yields the S-polynomial $g_m\cdot y^n x - xy^m \cdot g_n$.
Since both cases are completely analogous, let us present a complete proof only for the second one:
    \begin{align}
    \begin{aligned}
	   s_{m,n} \;&= &&(a_m y^m x y^n x - b_m xy^{m+n}x - y^{m+n}x) \\
            &&-&(a_n xy^{m+n}x - b_n xy^m xy^n - xy^{m+n}) \\
                        &= &-&(a_n + b_m)xy^{m+n}x + a_m y^m x y^n x \\
                        &&+&\,b_n xy^m xy^n - y^{m+n} x + xy^{m+n}.
    \end{aligned}\label{eq:spoly}
    \end{align}
    The terms $a_m y^m x y^n x$ and $b_n xy^m xy^n$ appearing in $s_{m,n}$ can be reduced 
    by $g_n$ and $g_m$, respectively. 
    This yields
    \begin{align*}
        a_m y^m x y^n x \;&\rightarrow_{g_n}\; a_m \cdot \left( -a_n y^{m+n} x + b_n y^m x y^n + y^{m+n}\right),\\
        b_n xy^m xy^n \;&\rightarrow_{g_m}\; b_n \cdot \left( -a_m y^m x y^n + b_m x y^{m+n} + y^{m+n}\right).
    \end{align*}
    Plugging this into~\eqref{eq:spoly} shows that $s_{m,n}$ can be reduced by $\mathcal G$ to 
    \begin{align*}
        s_{m,n} \;\xrightarrow{*}_G\; &-(a_n + b_m)xy^{m+n}x - (a_m a_n + 1) y^{m+n}x \\
                                            & + (b_m b_n + 1) xy^{m+n} + (a_m + b_n) y^{m+n}.
    \end{align*}
    By the last three items of Lemma~\ref{lemma:fib-properties}, the above polynomial is
    \[
        -(a_n + b_m) \cdot g_{m+n},
    \]
    which implies that it, and thus also $s_{m,n}$, can be reduced to zero by $\mathcal G$.
\end{example}

\end{document}